\documentclass[envcountsame]{llncs}
\usepackage[utf8]{inputenc}
\usepackage{amssymb}
\usepackage{amsmath}
\usepackage{bm}
\usepackage{stmaryrd}
\usepackage{textcomp}
\usepackage{proof}
\usepackage{color}
\usepackage{url}
\usepackage{cleveref}
\usepackage{multicol}
\newlength{\saut}

\newcommand{\alphabold}{\bm {\kappa}}

\newcommand{\cbold}{{\mathbf{k}}}
\newcommand{\tmu}{{\tilde\mu}}
\newcommand{\coupesep}{|\!|}
\newcommand{\coupe}[2]{
    {\mbox{$\langle$}} {#1} {\coupesep} {#2} \mbox{$\rangle$}}
\newcommand{\cut}[2]{\coupe{#1}{#2}}
\newcommand{\imp}{\rightarrow}
\newcommand{\defeq}{\triangleq}
\newcommand{\automath}[1]{\relax\ifmmode{#1}\else{$#1$}\fi}

\newcommand{\lbvtstar}{\automath{\overline{\lambda}_{[\lowercase{lv}\tau\star]}}}
\newcommand{\lbv} {\automath{\overline{\lambda}_{\lowercase{lv}}}}

\newcommand{\catch}{\texttt{catch}}
\newcommand{\throw}{\texttt{throw}}

\makeatletter
\newcommand{\oset}[3][0ex]{%
  \mathrel{\mathop{#3}\limits^{
    \vbox to#1{\kern-2\ex@
    \hbox{#2}\vss}}}}
\makeatother
\newcommand{\red}{\rightarrow}

\newcommand{\negt}{{\bot\!\!\!\!\!\bot}}

\newcommand{\indpt}[2]{{#1\,\#\,#2}}
\newcommand{\compat}[2]{{#1\diamond#2}}
\newcommand{\stext}{\vartriangleleft}
\newcommand{\join}[2]{\mathsf{join}(#1,#2)}
\newcommand{\C}{\mathcal{C}}
\renewcommand{\S}{\mathcal{S}}
\newcommand{\pole}{{\bot\!\!\!\bot}}
\newcommand{\orth}{{\pole}}
\newcommand{\prfcasetext}{Case}
\newcommand{\prfcase}[1]{\paragraph*{\normalfont \textbf{\prfcasetext} #1.}}
\newcommand{\prfcases}[1]{\paragraph*{\normalfont \textbf{{\prfcasetext}s} #1.}}

\newcommand{\fvn}[2]{\|#2\|_{#1}}
\newcommand{\tvn}[2]{|#2|_{#1}}
\newcommand{\fvF}[1]{\fvn{F}{#1}}
\newcommand{\fvE}[1]{\fvn{E}{#1}}
\newcommand{\fve}[1]{\fvn{e}{#1}}
\newcommand{\tvv}[1]{\tvn{v}{#1}}
\newcommand{\tvV}[1]{\tvn{V}{#1}}
\newcommand{\tvt}[1]{\tvn{t}{#1}}
\newcommand{\real}{\Vdash}

\newcounter{rep}
\newcounter{nrep}

\makeatletter

\spnewtheorem{repthm}[rep]{Theorem}{\bfseries}{\rmfamily}
 {\end{repthm}}

\spnewtheorem{nthm}[nrep]{Theorem}{\bfseries}{\rmfamily}
\newenvironment{ntheorem}[1]{%
\setcounter{nrep}{#1}
\addtocounter{nrep}{-1}
 \begin{nthm}}%
 {\end{nthm}}

\spnewtheorem{replm}[rep]{Lemma}{\bfseries}{\rmfamily}
 {\end{replm}}

\spnewtheorem{repprop}[rep]{Proposition}{\bfseries}{\rmfamily}
 {\end{repprop}}
\makeatother

\newenvironment{mycenter}{%
  \setlength\topsep{1pt}
  \setlength\parskip{1pt}
  \begin{center}
}{%
  \end{center}
}

\newcommand{\autorule}[1]{\relax\ifmmode{\scriptstyle(#1)}\else$(#1)$\fi}

\newcommand{\crule}    {\autorule{c     }}
\newcommand{\lrule}    {\autorule{l     }}
\newcommand{\xrule}    {\autorule{x     }}
\newcommand{\alpharule}{\autorule{\alpha}}
\newcommand{\cboldrule}    {\autorule{    \cbold}}
\newcommand{\alphaboldrule}{\autorule{\alphabold}}
\newcommand{\lft}[1]{\uparrow^{#1}}
\newcommand{\liftVrule}{\autorule{\lft{V}}}
\newcommand{\liftErule}{\autorule{\lft{E}}}
\newcommand{\lifttrule}{\autorule{\lft{t}}}
\newcommand{\lifterule}{\autorule{\lft{e}}}

\newcommand{\murule   }{\autorule{\mu      }}
\newcommand{\mutrule  }{\autorule{\tmu     }}
\newcommand{\eagerrule}{\autorule{\tmu^{[]}}}

\newcommand{\imprrule}{\autorule{\imp_r         }}
\newcommand{\implrule}{\autorule{\imp_l         }}

\newcommand{\ffalllrule}{\autorule{\forall_l^1      }}
\newcommand{\ffallrrule}{\autorule{\forall_r^1      }}
\newcommand{\sfalllrule}{\autorule{\forall_l^2      }}
\newcommand{\sfallrrule}{\autorule{\forall_r^2      }}

\newcommand{\epsrule    }{\autorule{\varepsilon }}
\newcommand{\tautrule   }{\autorule{\tau_t      }}
\newcommand{\tauErule   }{\autorule{\tau_E      }}

\newcommand{\hsep}{\qquad\qquad}
\newcommand{\lmmt}{\lambda\mu\tmu}
\newcommand{\vlong}[1]{#1}

\begin{document}
\title{Realizability Interpretation And Normalization Of Typed Call-by-Need $\lambda$-calculus With Control}
\author{\'Etienne Miquey\inst{1,2} \and Hugo Herbelin\inst{2}}
           
\institute{
Équipe Gallinette\\
Inria, LS2N (CNRS)\\
Université de Nantes\\
\and
Équipe $\pi r^2$\\
Inria, IRIF (CNRS)\\
Université Paris-Diderot\\
\email{{etienne.miquey,herbelin}@inria.fr}
}

\maketitle

\begin{abstract}
  We define a variant of Krivine realizability where realizers are pairs of a
  term and a substitution. This variant allows us to prove the
  normalization of a simply-typed call-by-need $\lambda$-calculus with
  control due to Ariola {\em et al}. Indeed, in such call-by-need
  calculus, substitutions have to be delayed until knowing if an
  argument is really needed.
  We then extend the proof to a call-by-need
  $\lambda$-calculus equipped with a type system equivalent to
  classical second-order predicate logic, representing one step
  towards proving the normalization of the call-by-need classical
  second-order arithmetic introduced by the second author to provide
  a proof-as-program interpretation of the axiom of dependent
  choice.
\end{abstract}

\section*{Introduction}

\subsection*{Realizability-based normalization}

Normalization by realizability 
is a standard technique to prove the normalization of typed $\lambda$-calculi. 
Originally introduced by Tait~\cite{Tait67} to prove
the normalization of System T, it was extended by Girard to prove the
normalization of System~F~\cite{Girard71}.
This kind of techniques, also called normalization by reducibility 
or normalization by logical relations,
works by interpreting each type by a set of typed or untyped terms seen as realizers of
the type, then showing that the way these sets of realizers are built preserve
properties such as normalization. 
Over the years, multiple uses and generalization of this method have been
done, for a more detailed account of which
we refer the reader to the work of Gallier~\cite{Gallier90}.

Realizability techniques were adapted to the normalization 
of various calculi for classical logic (see e.g. \cite{BarBer96a,Parigot00}). A specific framework
tailored to the study of realizability for classical logic has been designed by
Krivine~\cite{Krivine04} on top of a $\lambda$-calculus with control
whose reduction is defined in terms of an abstract machine. 
In such a machinery, terms are evaluated in front of stacks; and control (thus classical logic)
is made available through the possibility of saving and restoring stacks.
During the last twenty years, Krivine's classical realizability turned out to be
fruitful both from the point of view of logic, leading to the construction of 
new models of set theory, and generalizing in particular the technique of Cohen's forcing~\cite{Krivine11,Krivine12,Krivine14};
and on its computational facet, providing alternative tools to
the analysis of the computational content of classical programs\footnote{See for instance
\cite{Miquel11} about witness extraction or \cite{GuiMiq16,GuiM16} 
about specification problems.}.

Noteworthily, Krivine realizability is one of the approaches contributing to 
advocating the motto
that through
the Curry-Howard correspondence, with new programming instructions come 
new reasoning principles\footnote{For instance, one way to realize
the axiom of dependent choice in classical realizability is
by means of an extra instruction \texttt{quote}~\cite{Krivine03}.}.
Our original motivation for the present work is actually in line with this idea,
in the sense that our long-term purpose is to give a realizability interpretation
to $\text{dPA}^\omega$, a call-by-need calculus defined by the second 
author~\cite{Herbelin12}. 
In this calculus, the lazy evaluation is indeed a fundamental ingredient in order to 
obtain an executable proof term for the axiom of dependent choice.

\subsection*{Contributions of the paper}
In order to address the normalization of typed call-by-need $\lambda$-calculus,
we design a variant of Krivine's classical realizability, where
the realizers 
are closures (a term with a substitution for its free variables).
The call-by-need $\lambda$-calculus with control that we consider is the $\lbvtstar$-calculus. 
This calculus, that was defined by Ariola {\em et al.}~\cite{AriDowHerNakSau12},
is syntactically described in an extension with explicit substitutions of the
$\lmmt$-calculus~\cite{CurHer00,HerbelinHdR,Munch09}.
The syntax of the $\lmmt$-calculus itself refines the syntax of the $\lambda$-calculus 
by syntactically distinguishing between \emph{terms} and \emph{evaluation contexts}.
It also contains \emph{commands} which combine terms and
evaluation contexts so that they can interact together.
Thinking of evaluation contexts as stacks and commands as states, 
the $\lmmt$-calculus can also be seen as a syntax for abstract machines. 
As for a proof-as-program point of view, the $\lmmt$-calculus and its variants 
can be seen as a term syntax for proofs of Gentzen's sequent calculus. 
In particular, the $\lmmt$-calculus contains control operators 
which give a computational interpretation to classical logic.


We give a proof of normalization first for the simply-typed $\lbvtstar$-calculus\footnote{Even though
it has not been done formally, the normalization of the $\lbv$-calculus presented in~\cite{AriDowHerNakSau12}
should also be derivable from Polonowski's proof of strong normalization of the non-deterministic $\lmmt$-calculus~\cite{Polonowski04}. 
The $\lbv$-calculus (a big-step variant of the $\lbvtstar$-calculus introduced in Ariola {\em et al.}) is indeed a particular 
evaluation strategy for the $\lmmt$-calculus, so that the strong normalization of the non-deterministic variant of the latter 
should imply the normalization of the former as a particular case.}, then for a type system with first-order 
and second-order quantification.
While we only apply our technique to the normalization of the \lbvtstar-calculus, 
our interpretation incidentally suggests a way to adapt Krivine realizability to other call-by-need settings. 
This paves the way to the computational interpretation of classical
proofs using lazy evaluation or shared memory cells,
including the case of the call-by-need second order arithmetic $\text{dPA}^\omega$~\cite{Herbelin12}.

%
%
%
%

\section{The {\lbvtstar}-calculus}
\label{s:lbvtstar}
\subsection{The call-by-need evaluation strategy}
The call-by-need evaluation strategy of
the $\lambda$-calculus evaluates arguments of functions 
only when needed, and, when needed, shares their evaluations across
all places where the argument is required. 
The call-by-need evaluation is at the heart of a functional
programming language such as Haskell. It has in common with
the call-by-value evaluation strategy that all places
where a same argument is used share the same value.  
Nevertheless, it observationally behaves like the 
call-by-name evaluation strategy (for the pure $\lambda$-calculus),
in the sense that a given computation
eventually evaluates to a value if and only if it evaluates to the same
value (up to inner reduction) along the call-by-name evaluation.
In particular, in a setting with non-terminating computations, it is not observationally equivalent to the call-by-value
evaluation. Indeed, if the evaluation of a useless argument loops in
the call-by-value evaluation, the whole computation loops, which is not
the case of call-by-name and call-by-need evaluations.

These three evaluation strategies can be turned into equational theories.
For call-by-name and call-by-value, this was done by Plotkin through
continuation-passing-style (CPS) semantics characterizing these theories~\cite{Plotkin75}.  
For the call-by-need evaluation strategy, a specific equational theory
reflecting the intensional behavior of the strategy into a semantics was proposed independently by
Ariola and Felleisen~\cite{AriFel93}, and by Maraist, Odersky and Wadler~\cite{MarOdeWad98}.
A continuation-passing-style semantics was proposed in the 90s
by Okasaki, Lee and Tarditi~\cite{OkaLeeTar94}.
However, this semantics does not ensure normalization of simply-typed call-by-need
evaluation, as shown in~\cite{AriDowHerNakSau12}, thus failing to
ensure a property which holds in the simply-typed call-by-name
and call-by-value cases.

Continuation-passing-style semantics \emph{de facto} gives a semantics to the
extension of $\lambda$-calculus with control operators\footnote{That is to say with operators
such as Scheme's {\tt callcc}, Felleisen's ${\mathcal C}$, ${\mathcal K}$, or
${\mathcal A}$ operators~\cite{FelFriKohDub86}, Parigot's $\mu$ and [~]
operators~\cite{Parigot91}, Crolard's $\catch$ and $\throw$
operators~\cite{Crolard99}.}. In particular, even though call-by-name
and call-by-need are observationally equivalent on pure
$\lambda$-calculus, their different intentional behaviors induce
different CPS semantics, leading to different observational behaviors when control operators are
considered.
On the other hand, the semantics of calculi with control can also be
reconstructed from an analysis of the duality between programs and
their evaluation contexts, and the duality between the {\tt let}
construct (which binds programs) and a control operator such as
Parigot's $\mu$ (which binds evaluation contexts). 
Such an analysis can be done in the context of the
$\lmmt$-calculus~\cite{CurHer00,HerbelinHdR}.

In the call-by-name and call-by-value cases, the approach based on
$\lmmt$-calculus leads to continuation-passing style semantics
similar to the ones given by Plotkin or, in the call-by-name case, also to
the one by Lafont, Reus and Streicher~\cite{LafReuStr93}.
As for call-by-need,
in~\cite{AriDowHerNakSau12} is defined the $\lbv$-calculus, a call-by-need version
of the $\lmmt$-calculus. A continuation-passing style semantics is then defined via a
calculus called $\lbvtstar$~\cite{AriDowHerNakSau12}. 
This semantics, which is different from Okasaki, Lee and Tarditi's one~\cite{OkaLeeTar94}, 
is the object of study in this paper.

\subsection{Explicit environments}
While the results presented in this paper
could be directly expressed using the $\lbv$-calculus,
the realizability interpretation naturally arises from the decomposition of this calculus into 
a different calculus with an explicit \emph{environment}, the $\lbvtstar$-calculus~\cite{AriDowHerNakSau12}.
Indeed, as we shall see in the sequel, the decomposition highlights different syntactic categories 
that are deeply involved in the type system and in the definition of the realizability interpretation.

The $\lbvtstar$-calculus is a reformulation of the $\lbv$-calculus with explicit environments,
called  \emph{stores} and denoted by $\tau$.
Stores consists of a list of bindings of the form $[x:=t]$, where $x$ is a term variable and $t$ a term,
and of bindings of the form $[\alpha:=e]$ where $\alpha$ is a context variable and $e$ a context. 
For instance, in the closure $c\tau[x:=t]\tau'$, the variable $x$ is bound to $t$ in $c$ and $\tau'$.
Besides, the term $t$ might be an unevaluated term (\emph{i.e.} lazily stored),
so that if $x$ is eagerly demanded at some point during the execution of this closure, $t$ will be reduced in order to obtain a value.
In the case where $t$ indeed produces a value $V$, the store will be updated with the binding $[x:=V]$. 
However, a binding of this form (with a value) is fixed for the rest of the execution.
As such, our so-called stores somewhat behave like lazy explicit substitutions or mutable 
environments.

To draw the comparison
between our structures and the usual notions of stores and environments, two things should be observed.
First, the usual notion of store refers to a structure of list that is fully mutable, in the sense that the cells can be updated at any time and thus values might be replaced. 
Second, the usual notion of environment designates a structure in which variables are bounded to closures made of a term and an environment. 
In particular, terms and environments are duplicated, \emph{i.e.} sharing is not allowed. Such a structure resemble to a tree whose nodes
are decorated by terms, as opposed to a machinery allowing sharing (like ours) whose underlying structure is broadly a directed acyclic graph.
See for instance~\cite{Lang07} for a Krivine abstract machine with sharing.

\subsection{Syntax \& reduction rules}
\label{sec:syntax}

The lazy evaluation of terms allows for the following reduction
rule:
us to reduce a command $\cut{\mu\alpha .c}{\tmu x.c'}$ to the command $c'$ together with the binding $[x:=\mu\alpha.c]$.
$$\cut{\mu\alpha .c}{\tmu x.c'} \red c'[x:=\mu\alpha.c]$$
In this case, the term $\mu\alpha.c$ is left unevaluated (``frozen'') in the store, until possibly reaching a command in which the variable $x$ is needed.
When evaluation reaches a command of the form $\cut{x}{F}\tau[x:=\mu\alpha.c]\tau'$, the binding is opened and the term is evaluated in front of the context $\tmu[x].\cut{x}{F}\tau'$:
 $$\cut{x}{F}\tau[x:=\mu\alpha.c]\tau'\red\cut{\mu\alpha.c}{\tmu[x].\cut{x}{F}\tau'}\tau$$
The reader can think of the previous rule as the ``defrosting'' operation of the frozen term $\mu\alpha.c$\,:
this term is evaluated in the prefix of the store $\tau$ which predates it, in front of the context $\tmu[x].\cut{x}{F}\tau'$
where the $\tmu[x]$ binder is waiting for a value. 
This context keeps trace of the part of the store $\tau'$ that was originally located after the binding $[x:=...]$.
This way, if a value $V$ is indeed furnished for the binder $\tmu[x]$, the original command $\cut{x}{F}$ is evaluated in the updated full store:
 $$\cut{V}{\tmu[x].\cut{x}{F}\tau'}\tau \red \cut{V}{F}\tau[x:=V]\tau'$$
The brackets in $\tmu[x].c$ are used 
to express the fact that the variable $x$ is forced at top-level (unlike contexts of the shape $\tmu x.C[\cut{x}{F}]$ in the $\lbv$-calculus).
The reduction system resembles the one of an abstract machine. 
Especially, it allows us to keep the standard redex at the top of a command and avoids searching through the
meta-context for work to be done.

Note that our approach slightly differ from \cite{AriDowHerNakSau12}
since we split values into two categories: strong values ($v$) and
weak values ($V$). The strong values correspond to values strictly
speaking. The weak values include the variables which force the
evaluation of terms to which they refer into shared strong
value. Their evaluation may require capturing a continuation.
The syntax of the language, which includes
constants $\cbold$ and co-constants $\alphabold$,
is given in Figure \ref{fig:reduction-rules}.
As for the reduction $\rightarrow$,
we define it as the compatible reflexive transitive closure of the rules
  given in Figure \ref{fig:reduction-rules}.
 {
\begin{figure}[t]
\framebox{\vbox{
$$
\begin{array}{l@{\qquad}rll}
\mbox{Strong values} 	& v   & ::= & \lambda x.t \mid \cbold   \\
\mbox{Weak values} 	& V   & ::= & v \mid  x 		\\
\mbox{Terms}		& t,u & ::= & V \mid  \mu\alpha.c 	\\
\mbox{Forcing contexts}   & F & ::= & t \cdot E \mid \alphabold \\
\mbox{Catchable contexts}  & E & ::= & F \mid  \alpha \mid  \tmu[x].\cut{x}{F}\tau\\
\mbox{Evaluation contexts} & e & ::= & E \mid  \tmu x.c\\
\mbox{Stores} & \tau & ::= & \varepsilon \mid  \tau[x:=t] \mid   \tau[\alpha:=E]\\
\mbox{Commands} & c & ::= & \cut{t}{e} \\ 
\mbox{Closures} & l & ::= & c\tau \\ 
\end{array}$$
\hrule
$$
\begin{array}{c@{\qquad\rightarrow\qquad}c}
\cut{\lambda x.t}{u\cdot E}\tau      & \cut{u}{\tmu x.\cut{t}{E}}\tau \\
\cut{t}{\tmu x.c }\tau               & c\tau[x:=t] \\
\cut{\mu\alpha.c}{E}\tau             & c\tau[\alpha:=E] \\
\cut{V}{\alpha}\tau[\alpha:=E]\tau'  & \cut{V}{E}\tau[\alpha:=E]\tau' \\
\cut{x}{F}\tau[x:=t]\tau'            & \cut{t}{\tmu[x].\cut{x}{F}\tau'}\tau \\
\cut{V}{\tmu[x].\cut{x}{F}\tau'}\tau & \cut{V}{F}\tau[x:=V]\tau' 
\end{array}
\leqno
\begin{array}{l}	
 (\textsc{Beta}         )\\
 (\textsc{Let}          )\\
 (\textsc{Catch}        )\\
 (\textsc{Lookup}_\alpha)\\
 (\textsc{Lookup}_x     )\\
 (\textsc{Restore}      )\\
\end{array}
$$
}}
\caption{Syntax and reduction rules of the $\lbvtstar$-calculus}

\label{fig:reduction-rules}
\end{figure}
}

The different syntactic categories can be understood as the different
levels of alternation in a context-free abstract 
machine (see~\cite{AriDowHerNakSau12}):
the priority is first given to contexts at level $e$ (lazy storage of terms),
then to terms at level $t$ (evaluation of $\mu\alpha$ into values),
then back to contexts at level $E$ and so on until level~$v$.
These different categories are directly reflected in the definition of the 
abstract machine defined in~\cite{AriDowHerNakSau12},
and will thus be involved in the definition of our realizability interpretation.
We chose to highlight this by distinguishing different types of sequents
already in the typing rules that we shall now present.

\subsection{A type system for the $\lbvtstar$-calculus.}
\label{sec:typing}
\begin{figure}[t]
  \framebox
{\vbox{
\input{figures/typing_rules_lbvtstar}
}}
\caption{Typing rules of the $\lbvtstar$-calculus}
\label{fig:typing-rules}
\end{figure}


We have nine kinds of (one-sided) sequents, one for typing each of the nine
syntactic categories. We write them with an annotation on the $\vdash$
sign, using one of the letters $v$, $V$, $t$, $F$, $E$, $e$, $l$, $c$,
$\tau$.
Sequents typing values and
terms are asserting a type, with the type written on the right;
sequents typing contexts are expecting a type $A$ with the type written
$A^\negt$; sequents typing commands and closures are black boxes neither 
asserting nor expecting a type; sequents typing substitutions are 
instantiating a typing context.
In other words, we have the following nine kinds of sequents:
\begin{center}
\begin{tabular}{l@{\qquad\quad}l@{\qquad\quad}l}
\begin{tabular}{l}
$\Gamma \vdash_l l$\\
$\Gamma \vdash_c c$\\
$\Gamma \vdash_\tau \tau:\Gamma' $\\
\end{tabular}
&
\begin{tabular}{l}
$\Gamma \vdash_t t:A $\\
$\Gamma \vdash_V V:A $\\
$\Gamma \vdash_v v:A $\\
\end{tabular}
&
\begin{tabular}{l}
$\Gamma \vdash_e e:A^\negt $\\
$\Gamma \vdash_E E:A^\negt $\\
$\Gamma \vdash_F F:A^\negt $\\
\end{tabular}
\end{tabular}

\end{center}
where types and typing contexts are defined by:
$$
 A,B  ::=  X \mid  A \imp B \qquad\qquad\qquad
 \Gamma ::=  \varepsilon \mid  \Gamma, x:A\mid \Gamma,\alpha:A^\negt
$$

The typing rules are given on Figure \ref{fig:typing-rules} where we assume 
that a variable $x$ (resp. co-variable $\alpha$) only occurs once in a context $\Gamma$
(we implicitly assume the possibility of renaming variables by $\alpha$-conversion).
We also adopt the convention that constants $\cbold$ and co-constants $\alphabold$ come with a signature $\S$
which assigns them a type.
This type system enjoys the property of subject reduction.


\begin{theorem}[Subject reduction]
\label{thm:subject}
 If $\Gamma \vdash_l c\tau$ and $c\tau\rightarrow c'\tau'$ then $\Gamma\vdash_l c'\tau'$.
\end{theorem}
\begin{proof}
 By induction on typing derivations\vlong{, see Appendix \ref{s:sub_red}}.\qed
\end{proof}

\section{Normalization of the {\lbvtstar}-calculus}
\label{s:real}
\renewcommand{\mathsf}[1]{\texttt{#1}}
\newcommand{\dom}{\mathsf{dom}}
\newcommand{\ct}{term-in-store}
\newcommand{\cts}{terms-in-store}
\newcommand{\ce}{context-in-store}

\subsection{Normalization by realizability}
The proof of normalization for the $\lbvtstar$-calculus that we present in this section 
is inspired from techniques of Krivine's classical realizability~\cite{Krivine04}, whose notations we borrow.
Actually, it is also very close to a proof by reducibility\footnote{See for instance the proof of normalization for system $D$ presented in~\cite[3.2]{Krivine93}.}.
In a nutshell, to each type $A$ is associated a set $|A|_t$ of terms whose execution is guided by the structure of $A$.
These terms are the ones usually called \emph{realizers} in Krivine's classical realizability.
Their definition is in fact indirect, and is done by orthogonality to a set of ``correct'' computations,
called a \emph{pole}. The choice of this set is central when studying models induced by classical realizability
for second-order-logic, but in the present case we only pay attention to the particular pole of terminating computations.
This is where lies one of the difference with usual proofs by reducibility, 
where everything is done with respect to $SN$, while our definition are parametric in the pole (which is chosen to be $SN$ in the end).
The adequacy lemma, which is the central piece,
consists in proving that typed terms belong to the corresponding sets of realizers, and are thus normalizing.

More in details, our proof can be sketched as follows.
First, we generalize the usual notion of closed term to the notion of closed \emph{\ct}.
Intuitively, this is due to the fact that we are no longer interested in closed terms and substitutions to close opened terms, 
but rather in terms that are closed when considered in the current store. 
This is based on the simple observation that a store is nothing more than a shared substitution whose 
content might evolve along the execution.
Second, we define the notion of \emph{pole} $\pole$, which are sets of closures closed by anti-evaluation and store extension.
In particular, the set of normalizing closures is a valid pole.
This allows to relate terms and contexts thanks to a notion of orthogonality with respect to the pole. 
We then define for each formula $A$ and typing level $o$ (of $e,t,E,V,F,v$) 
a set $|A|_o$ (resp. $\|A\|_o$) of terms (resp. contexts) in the corresponding syntactic category.
These sets correspond to reducibility candidates, or to what is usually called truth values and falsity values in Krivine realizability.
Finally, the core of the proof consists in the adequacy lemma, 
which shows that any closed term of type $A$ at level $o$ is in the corresponding set $|A|_o$.
This guarantees that any typed closure is in any pole, and in particular in the pole of normalizing closures.
Technically, the proof of adequacy evaluates in each case a state of an abstract machine (in our case a closure), 
so that the proof also proceeds by evaluation.
A more detailed explanation of this observation as well as a more introductory presentation of normalization proofs by classical realizability
are given in an article by Dagand and Scherer~\cite{DagSch15}.

\subsection{Realizability interpretation for the $\lbvtstar$-calculus}
\par
We begin by defining some key notions for stores that we shall need further in the proof.
\begin{definition}[Closed store]
 We extend the notion of free variable to stores:\vspace{-0.4em}
 $$\begin{array}{c@{~~\defeq~~}l}
 FV(\varepsilon) 	& \emptyset \\
 FV(\tau[x:=t]) 	& FV(\tau)\cup\{y\in FV(t):y\notin \dom(\tau)\} \\
 FV(\tau[\alpha:=E])	& FV(\tau)\cup\{\beta \in FV(E):\beta \notin \dom(\tau)\}
 \end{array}$$\vspace{-0.4em}
so that we can define a \emph{closed store} to be a store $\tau$ such that $FV(\tau) = \emptyset$.
\end{definition}

\begin{definition}[Compatible stores]
 We say that two stores $\tau$ and $\tau'$ are \emph{independent} and write $\indpt{\tau}{\tau'}$ when
 ${\dom(\tau)\cap\dom(\tau')=\emptyset}$.
 We say that they are \emph{compatible} and write $\compat{\tau}{\tau'}$ 
 whenever for all variables $x$ (resp. co-variables $\alpha$) present in both stores: ${x\in \dom(\tau)\cap\dom(\tau')}$;
 the corresponding terms (resp. contexts) in $\tau$ and $\tau'$ 
 coincide.
 Finally, we say that $\tau'$ is an \emph{extension} of $\tau$ and write $\tau\stext \tau'$ whenever 
 $\dom(\tau)\subseteq\dom(\tau')$ and $\compat{\tau}{\tau'}$.
 \end{definition}

 We denote by $\overline{\tau\tau'}$ the compatible union $\join{\tau}{\tau'}$ of closed stores $\tau$ and $\tau'$, defined by:\vspace{-0.4em}
 $$\begin{array}{r@{~~\defeq~~}l}
 \join{\tau_0[x:=t]\tau_1}{\tau'_0[x:=t]\tau'_1} & \tau_0\tau'_0[x:=t]\join{\tau_1}{\tau'_1} \\
  \join{\tau}{\tau'} & \tau\tau'                                                             \\
  \join{\varepsilon}{\tau} & \tau \\
 \join{\tau}{\varepsilon} &\tau\\
 \end{array}\eqno\begin{array}{r}(\text{if } \indpt{\tau_0}{\tau_0'})\\(\text{if }\indpt{\tau}{\tau'})\\\\\\\end{array} 
 $$
 
The following lemma (which follows easily from the previous definition) states the main property  
we will use about union of compatible stores.
\begin{lemma}
\label{lm:st_union}
If $\tau$ and $\tau'$ are two compatible stores, then $\tau\stext\overline{\tau\tau'}$ and $\tau'\stext\overline{\tau\tau'}$.
Besides, if $\tau$ is of the form $\tau_0[x:=t]\tau_1$,
then $\overline{\tau\tau'}$ is of the form ${\tau_2}[x:=t]{\tau_3}$ with $\tau_0 \stext {\tau_2}$
and $\tau_1\stext {\tau_3}$.
\end{lemma}
\begin{proof}
This follows easily from the previous definition.\qed
\end{proof}


\newcommand{\tis}[2]{(#1|#2)}
As we explained in the introduction of this section, we will not consider closed terms in the usual sense.
Indeed, while it is frequent in the proofs of normalization (\emph{e.g.} by realizability or reducibility) of a calculus to consider
only closed terms and to perform substitutions to maintain the closure of terms, this only makes sense if it corresponds to the 
computational behavior of the calculus. For instance, to prove the normalization of $\lambda x.t$ in
typed call-by-name $\lmmt$-calculus, one would consider a substitution $\rho$ that is suitable for 
with respect to the typing context $\Gamma$, then a context $u\cdot e$ of type $A\to B$, and evaluates :
$$\cut{\lambda x.t_\rho}{u\cdot e} ~~\rightarrow~~\cut{t_\rho[u/x]}{e}$$
Then we would observe that $t_\rho[u/x] = t_{\rho[x:=u]}$ and deduce that $\rho[x:=u]$ is suitable for $\Gamma,x:A$, 
which would allow us to conclude by induction.

However, in the $\lbvtstar$-calculus we do not perform global substitution when reducing a command, but rather add a new binding $[x:=u]$ in the store:
$$\cut{\lambda x.t}{u\cdot E}\tau ~~\rightarrow~~\cut{t}{E}\tau[x:=u]$$
Therefore, the natural notion of closed term invokes the closure under a store, 
which might evolve during the rest of the execution (this is to contrast with a substitution).

\begin{definition}[Term-in-store]
We call \emph{closed \ct} (resp. \emph{closed \ce}, \emph{closed closures}) 
the combination of a term $t$ (resp. context $e$, command $c$) with a closed store $\tau$ such that
$FV(t)\subseteq \dom(\tau)$. 
We use the notation $\tis{t}{\tau}$ (resp. $\tis{e}{\tau}, \tis{c}{\tau}$)
to denote such a pair. 
\end{definition}
We should note that in particular, if $t$ is a closed term, then $\tis{t}{\tau}$ is a {\ct} for any closed store $\tau$. 
The notion of {closed \ct} is thus a generalization of the notion of closed terms, and we will (ab)use of this 
terminology in the sequel. We denote the sets of closed closures by $\C_0$, and will identify $\tis{c}{\tau}$ and the closure $c\tau$ when $c$ is closed in $\tau$.
Observe that if $c\tau$ is a closure in $\C_0$ and $\tau'$ is a store extending $\tau$, then $c\tau'$ is also in $\C_0$.
We are now equipped to define the notion of pole, and verify that the set of normalizing closures is indeed a valid pole.
\begin{definition}[Pole]
 A subset $\pole\subseteq \C_0$ is said to be \emph{saturated} or \emph{closed by anti-reduction} 
 whenever for all $\tis{c}{\tau},\tis{c'}{\tau'}\in\C_0$,  if $c'\tau' \in \pole$ and $c\tau\rightarrow c'\tau'$ then $c\tau\in\pole$.
 It is said to be \emph{closed by store extension} if whenever $c\tau\in\pole$, for any store $\tau'$ extending $\tau$: $\tau\stext\tau'$, $c\tau'\in\pole$.
 A \emph{pole} is defined as any subset of $\C_0$ that is closed by anti-reduction and store extension.
\end{definition}

The following proposition is the one supporting the claim that our realizability proof is almost
a reducibility proof whose definitions have been generalized with respect to a pole instead of the fixed set SN.
\begin{proposition}
\label{prop:norm_pole}
 The set $\pole_{\Downarrow}=\{c\tau\in\C_0:~c\tau\text{ normalizes }\}$ is a pole.
\end{proposition}
\begin{proof}
 As we only considered closures in $\C_0$, both conditions (closure by anti-reduction and store extension) are clearly satisfied:
 \begin{itemize}
  \item if $c\tau \rightarrow c'\tau'$ and $c'\tau'$ normalizes, then $c\tau$ normalizes too;
  \item if $c$ is closed in $\tau$ and $c\tau$ normalizes, if $\tau\stext \tau'$ then $c\tau'$ will reduce as $c\tau$ does
  (since $c$ is closed under $\tau$, it can only use terms in $\tau'$ that already were in $\tau$) and thus will normalize.\qed
 \end{itemize}
\end{proof}

\begin{definition}[Orthogonality]
Given a pole $\pole$, we say that a {\ct} $\tis{t}{\tau}$ is {\em orthogonal} to a {\ce} $\tis{e}{\tau'}$
and write $\tis{t}{\tau}\orth\tis{e}{\tau'}$
if $\tau$ and $\tau'$ are compatible and $\cut{t}{e}\overline{\tau\tau'}\in\pole$.
\end{definition}
\begin{remark}
 The reader familiar with Krivine's forcing machine~\cite{Krivine11} might recognize his definition of orthogonality 
 between terms of the shape $(t,p)$ and stacks of the shape $(\pi,q)$, where $p$ and $q$ are forcing conditions\footnote{The meet of forcing conditions is indeed a refinement containing somewhat the ``union'' of information contained in each, 
 just like the union of two compatible stores.}:
 $$ (t,p) \pole (\pi,q) \Leftrightarrow (t\star\pi,p\land q) \in \pole$$
 
\end{remark}

We can now relate closed terms and contexts by orthogonality with respect to a given pole. 
This allows us to define for any formula $A$ 
the sets $\tvv{A},\tvV{A},\tvt{A}$ (resp. $\fvF{A}$,$\fvE{A}$, $\fve{A}$) 
of realizers (or reducibility candidates) at level $v$, $V$, $t$ (resp. $F$, $E$, $e$) for the formula $A$. 
It is to be observed that realizers are here closed {\cts}.

\begin{definition}[Realizers]
\label{def:realizers}
 Given a fixed pole $\pole$, we set:\vspace{-0.5em}
 $$\begin{array}{ccl}
     \tvv{X} 	 & = & \{\tis{\cbold}{\tau} : \quad\vdash \cbold:{X}\}\\
     \tvv{A\imp B} & = & \{\tis{\lambda x .t}{\tau} : \forall u \tau', \compat{\tau}{\tau'}\land \tis{u}{\tau'}\in\tvt{A} \Rightarrow \tis{t}{\overline{\tau\tau'}[x:=u]}\in\tvt{B}\}\\
     \fvF{A} 	 & = & \{\tis{F}{\tau} : \forall v \tau', \compat{\tau}{\tau'}\land \tis{v}{\tau'}\in\tvv{A} \Rightarrow \tis{v}{\tau'}\orth \tis{F}{\tau}\}\\
     \tvV{A} 	 & = & \{\tis{V}{\tau} : \forall F \tau', \compat{\tau}{\tau'}\land \tis{F}{\tau'}\in\fvF{A} \Rightarrow \tis{V}{\tau} \orth \tis{F}{\tau'}\}\\
     \fvE{A} 	 & = & \{\tis{E}{\tau} : \forall V \tau', \compat{\tau}{\tau'}\land \tis{V}{\tau'}\in\tvV{A} \Rightarrow \tis{V}{\tau'}\orth \tis{E}{\tau}\}\\
     \tvt{A} 	 & = & \{\tis{t}{\tau} : \forall E \tau', \compat{\tau}{\tau'}\land \tis{E}{\tau'}\in\fvE{A} \Rightarrow \tis{t}{\tau} \orth \tis{E}{\tau'}\}\\
     \fve{A} 	 & = & \{\tis{e}{\tau} : \forall t \tau', \compat{\tau}{\tau'}\land \tis{t}{\tau'}\in\tvt{A} \Rightarrow \tis{t}{\tau'}\orth \tis{e}{\tau}\}\\

   \end{array}$$
\end{definition}
\begin{remark}
 We draw the reader attention to the fact that we should actually write $\tvv{A}^\pole,\fvF{A}^\pole$, etc...
and $\tau\real_{\!\!\pole}\!\Gamma$, because the corresponding definitions 
are parameterized by a pole $\pole$. 
As it is common in Krivine's classical realizability, we ease the notations by
removing the annotation $\pole$ whenever there is no ambiguity on the pole.
Besides, it is worth noting that if co-constants do not occur directly in the definitions, they
 may still appear in the realizers by mean of the pole.
\end{remark}

If the definition of the different sets might seem complex at first sight, 
we claim that they are quite natural in regards of the methodology of Danvy's semantics artifacts presented in~\cite{AriDowHerNakSau12}.
Indeed, having an abstract machine in context-free form (the last step in this methodology before deriving the CPS)
allows us to have both the term and the context (in a command) that behave independently of each other.
Intuitively, a realizer at a given level is precisely a term which is going to behave well (be in the pole)
in front of any opponent chosen in the previous level (in the hierarchy $v,F,V$,etc...). 
For instance, in a call-by-value setting, there are only three levels of definition (values, contexts and terms) in the interpretation,
because the abstract machine in context-free form also has three. 
Here the ground level corresponds to strong values, and the other levels are somewhat defined as terms (or context) 
which are well-behaved in front of any opponent in the previous one. 
The definition of the different sets $\tvv{A},\fvF{A},\tvV{A}$, etc... directly stems from this intuition.

In comparison with the usual definition of Krivine's classical realizability, we only considered 
orthogonal sets restricted to some syntactical subcategories.
However, the definition still satisfies the usual monotonicity properties of bi-orthogonal sets:
\begin{proposition}
\label{prop:monotonicity}
For any type $A$ and any given pole $\pole$, we have:
\vspace{-0.5em}
\begin{multicols}{2}
   \begin{enumerate}
   \item $\tvv{A}\subseteq \tvV{A} \subseteq \tvt{A}$;
   \item $\fvF{A}\subseteq \fvE{A} \subseteq \fve{A}$.
  \end{enumerate}
\end{multicols}  
\end{proposition}
\begin{proof}
 All the inclusions are proved in a similar way. 
 We only give the proof for $\tvv{A}\subseteq \tvV{A}$.
 Let $\pole$ be a pole and $\tis{v}{\tau}$ be in $\tvv{A}$.
 We want to show that $\tis{v}{\tau}$ is in $\tvV{A}$, 
 that is to say that $v$ is in the syntactic category $V$ (which is true),
 and that for any $\tis{F}{\tau'}\in\fvF{A}$ 
 such that $\compat{\tau}{\tau'}$, $\tis{v}{\tau}\orth\tis{F}{\tau'}$. 
 The latter holds by definition of $\tis{F}{\tau'}\in\fvF{A}$,
 since $\tis{v}{\tau}\in\tvv{A}$.\qed
\end{proof}

We now extend the notion of realizers to stores, by stating that a store $\tau$ realizes a context $\Gamma$
if it binds all the variables $x$ and $\alpha$ in $\Gamma$ to a realizer of the corresponding formula. 

\begin{definition}Given a closed store $\tau$ and a fixed pole $\pole$, 
we say that $\tau$ \emph{realizes} $\Gamma$, which we write\footnote{Once again, 
we should formally write $\tau\real_{\!\!\pole}\!\Gamma$ but we will omit the annotation by $\pole$ as often as possible.} $\tau \Vdash \Gamma$, if:
\begin{enumerate}
 \item for any $(x:A) \in\Gamma$, $\tau\equiv \tau_0[x:=t]\tau_1$ and $\tis{t}{\tau_0} \in \tvt{A}$
 \item for any $(\alpha:A^\negt) \in\Gamma$, $\tau\equiv \tau_0[\alpha:=E]\tau_1$ and $\tis{E}{\tau_0} \in \fvE{A}$\
 \end{enumerate}
 \label{def:store_real}
\end{definition}

In the same way than weakening rules (for the typing context) are admissible for each level of the typing system :
$$
\infer={\Gamma'\vdash_t t:A}{\Gamma\vdash_t t:A & \Gamma\subseteq\Gamma'} 		\quad\qquad
\infer={\Gamma'\vdash_e e:A^\negt}{\Gamma\vdash_e e:A^\negt & \Gamma\subseteq\Gamma'} 	\quad~~
\text{\raisebox{1em}{$\dots$}}\quad~~
\infer={\Gamma'\vdash_\tau \tau:\Gamma''}{\Gamma\vdash_\tau \tau:\Gamma'' & \Gamma\subseteq\Gamma'}$$
the definition of realizers is compatible with a weakening of the store.
\begin{lemma}[Store weakening]
\label{lm:st_weak}
 Let $\tau$ and $\tau'$ be two stores such that $\tau\stext\tau'$, let $\Gamma$ be a typing context and 
 let $\pole$ be a pole. The following statements hold:
 \begin{enumerate}
  \item $\overline{\tau\tau'} = \tau'$
  \item If ~$\tis{t}{\tau} \in \tvt{A}$~ for some closed term $\tis{t}{\tau}$ and type $A$, then ~$\tis{t}{\tau'}\in\tvt{A}$. 
  The same holds for each level $e,E,V,F,v$ of the typing rules.
  \item If ~$\tau\real \Gamma$~ then ~$\tau' \real \Gamma$.
 \end{enumerate}
\end{lemma}
\begin{proof}
\begin{enumerate}
 \item Straightforward from the definition of $\bar{\tau\tau'}$.
 \item
 This essentially amounts to the following observations. First, one remarks that if $\tis{t}{\tau}$ is a closed term, so 
 then so is $\tis{t}{\overline{\tau\tau'}}$ for any closed store $\tau'$ compatible with $\tau$.
 Second, we observe that if we consider for instance a closed context $\tis{E}{\tau''}\in\fvE{A}$,
 then $\compat{\overline{\tau\tau'}}{\tau''}$ implies $\compat{\tau}{\tau''}$, thus $\tis{t}{\tau}\orth\tis{E}{\tau''}$
 and finally $\tis{t}{\overline{\tau\tau'}}\orth\tis{E}{\tau''}$ by closure of the pole under store extension.
 We conclude that $\tis{t}{\tau'}\orth\tis{E}{\tau''}$ using the first statement.
 
 \item
 By definition, for all $(x:A)\in\Gamma$, $\tau$ is of the form $\tau_0[x:=t]\tau_1$ such that 
 $\tis{t}{\tau_0}\in\tvt{A}$.
 As $\tau$ and $\tau'$ are compatible, we know by Lemma~\ref{lm:st_union} that $\overline{\tau\tau'}$ is of the form $\tau'_0[x:=t]\tau'_1$ 
 with $\tau'_0$ an extension of $\tau_0$, and using the first point we get that $\tis{t}{\tau'_0}\in\tvt{A}$.\qed
 \end{enumerate}
\end{proof}

\begin{definition}[Adequacy]
Given a fixed pole~$\pole$, we say that:
\begin{itemize}
\item A typing judgment $\Gamma\vdash_t t:A$ is \emph{adequate}
  (w.r.t.\ the pole~$\pole$) if for all stores $\tau\real\Gamma$,
  we have $\tis{t}{\tau} \in \tvt{A}$.
\item More generally, we say that an inference rule
  $$\infer{J_0}{J_1 &\cdots& J_n}$$
  is adequate (w.r.t.\ the pole~$\pole$) if the adequacy of all typing
  judgments $J_1,\ldots,J_n$ implies the adequacy of the typing
  judgment $J_0$.
\end{itemize}
\end{definition}

\begin{remark}
\label{rmk:adequacy}
From the latter definition, it is clear that a typing judgment that is
derivable from a set of adequate inference rules is adequate too.
\end{remark}

We will now show the main result of this section, namely that the typing rules 
of Figure~\ref{fig:typing-rules} for the $\lbvtstar$-calculus without co-constants
are adequate with any pole. 
Observe that this result requires to consider the $\lbvtstar$-calculus without co-constants. 
Indeed, we consider co-constants as coming with their typing rules, 
potentially giving them any type (whereas constants can only be given an atomic type).
Thus, there is \emph{a priori} no reason\footnote{Think for instance of a co-constant of type $(A\imp B)^\negt$, 
there is no reason why it should be orthogonal to any function in $\tvv{A\imp B}$.}
why their types should be adequate with any pole. 

However, as observed in the previous remark, given a fixed pole it 
suffices to check whether the typing rules for a given co-constant
are adequate with this pole. If they are, any judgment that is derivable
using these rules will be adequate.

\begin{theorem}[Adequacy]\label{lm:adequacy}
If $\Gamma$ is a typing context, $\pole$ is a pole and $\tau$ is a store such that ${\tau\real \Gamma}$,
then the following holds in the $\lbvtstar$-calculus without co-constants:
\begin{enumerate}
 \item If $v$ is a strong value 	such that $\Gamma\vdash_v v:A$, 	then $\tis{v}{\tau} \in\tvv{A}$.
 \item If $F$ is a forcing context 	such that $\Gamma\vdash_F F:A^\negt$, 	then $\tis{F}{\tau} \in\fvF{A}$.
 \item If $V$ is a weak value   	such that $\Gamma\vdash_V V:A$, 	then $\tis{V}{\tau} \in\tvV{A}$.
 \item If $E$ is a catchable context 	such that $\Gamma\vdash_E E:A^\negt$, 	then $\tis{E}{\tau} \in\fvF{A}$.
 \item If $t$ is a term    		such that $\Gamma\vdash_t t:A$, 	then $\tis{t}{\tau} \in\tvt{A}$.
 \item If $e$ is a context 		such that $\Gamma\vdash_e e:A^\negt$, 	then $\tis{e}{\tau} \in\fve{A}$.
 \item If $c$ is a command		such that $\Gamma\vdash_c c$, 		then $c\tau \in \pole$. 
 \item If $\tau'$ is a store	such that $\Gamma\vdash_\tau \tau':\Gamma'$, 	then $\tau\tau' \real \Gamma,\Gamma'$. 
\end{enumerate}

\end{theorem}
\begin{proof}
\renewcommand{\prfcasetext}{Rule}
The different statements are proved by mutual induction over typing derivations.
We only give the most important cases here\vlong{, the exhaustive induction is given in Appendix~\ref{a:real}}.


%

\prfcase{\implrule}
Assume that $$\infer[\implrule]{\Gamma\vdash_F u\cdot E:(A\imp B)^\negt}{\Gamma\vdash_t u:A & \Gamma \vdash_E E:B^\negt}$$ 
and let $\pole$ be a pole and $\tau$ a store such that $\tau \real \Gamma$.
Let $\tis{\lambda x.t}{\tau'}$ be a closed term in the set $\tvv{A\to B}$ such that $\compat{\tau}{\tau'}$, then we have:
$$\cut{\lambda x.t}{u\cdot E}\overline{\tau\tau'}
~~\rightarrow ~~\cut{u}{\tmu x.\cut{t}{E}}\overline{\tau\tau'}
~~\rightarrow ~~\cut{t}{E}\overline{\tau\tau'}[x:=u]$$
By definition of $\tvv{A\to B}$, this closure is in the pole, and we can conclude by anti-reduction.


\prfcase{\xrule} 
Assume that 
$$\infer[\xrule]{\Gamma\vdash_V x:A}{(x:A)\in\Gamma}$$ 
and let $\pole$ be a pole and $\tau$ a store such that $\tau \real \Gamma$.
As $(x:A)\in\Gamma$, we know that $\tau$ is of the form $\tau_0[x:=t]\tau_1$
with $\tis{t}{\tau_0}\in\tvt{A}$.
Let $\tis{F}{\tau'}$ be in $\fvF{A}$, with  $\compat{\tau}{\tau'}$. By Lemma~\ref{lm:st_union},
we know that $\overline{\tau\tau'}$ is of the form $\overline{\tau_0}[x:=t]\overline{\tau_1}$.
Hence we have:
$$\cut{x}{F} \overline{\tau_0}[x:=t]\overline{\tau_1} ~~\rightarrow~~ \cut{t}{\tmu[x].\cut{x}{F}\overline{\tau_1}}\overline{\tau_0}$$
and it suffices by anti-reduction to show that the last closure is in the pole $\pole$. 
By induction hypothesis, we know that $\tis{t}{\tau_0}\in\tvt{A}$ thus we only need to show 
that it is in front of a catchable context in $\fvE{A}$.
This corresponds exactly to the next case that we shall prove now.

\prfcase{\eagerrule} 
Assume that
$$\infer[\eagerrule]{\Gamma\vdash_E \tmu[x].\cut{x}{F}\tau':A}{\Gamma,x:A,\Gamma'\vdash_F F:A & \Gamma,x:A\vdash \tau':\Gamma'}$$ 
and let $\pole$ be a pole and $\tau$ a store such that $\tau \real \Gamma$.
Let $\tis{V}{\tau_0}$ be a closed term in $\tvV{A}$ such that $\compat{\tau_0}{\tau}$.
We have that :
$$\cut{V}{\tmu[x].\cut{x}{F}\overline{\tau'}}\overline{\tau_0\tau}~~\rightarrow~~ \cut{V}{F} \overline{\tau_0\tau}[x:=V]\tau'$$
By induction hypothesis, we obtain $\tau[x:=V]\tau'\real \Gamma,x:A,\Gamma'$. 
Up to $\alpha$-conversion in $F$ and $\tau'$, so that the variables in $\tau'$ are disjoint from those in $\tau_0$,
we have that $\overline{\tau_0\tau}\real\Gamma$ (by Lemma~\ref{lm:st_weak}) and then $\tau''\defeq\overline{\tau_0\tau}[x:=V]\tau'\real \Gamma,x:A,\Gamma'$.
By induction hypothesis again, we obtain that $\tis{F}{\tau''}\in\fvF{A}$ 
(this was an assumption in the previous case) 
and as $\tis{V}{\tau_0}\in\tvV{A}$, we finally get that $\tis{V}{\tau_0}\orth\tis{F}{\tau''}$ 
and conclude again by anti-reduction.\qed

%
%
%
%

 \end{proof}

\begin{corollary}
 If $c\tau$ is a closure such that $\vdash_l c\tau$ is derivable, then for any pole $\pole$ such that 
 the typing rules for co-constants used in the derivation are adequate with $\pole$, $c\tau\in\pole$.
\end{corollary}

We can now put our focus back on the normalization of typed closures. 
As we already saw in Proposition~\ref{prop:norm_pole}, the set $\pole_{\Downarrow}$ of normalizing closure is a valid pole,
so that it only remains to prove that any typing rule for co-constants is adequate with $\pole_{\Downarrow}$.

\begin{lemma}
 Any typing rule for co-constants is adequate with the pole $\pole_{\Downarrow}$, \emph{i.e.} if $\Gamma$ is a typing context,
 and $\tau$ is a store such that $\tau\real\Gamma$,
 if $\alphabold$ is a co-constant such that $\Gamma\vdash_F \alphabold:A^\negt$, then $\tis{\alphabold}{\tau}\in\fvF{A}$.
\end{lemma}
 \begin{proof}
This lemma directly stems from the observation that for any store $\tau$ and any closed strong value $\tis{v}{\tau'}\in\tvv{A}$, 
 $\cut{v}{\alphabold}\overline{\tau\tau'}$ does not reduce and thus belongs to the pole $\pole_{\Downarrow}$.
 \end{proof}

As a consequence, we obtain the normalization of typed closures of the full calculus.
\begin{theorem}
\label{thm:normalization}
 If $c\tau$ is a closure of the $\lbvtstar$-calculus such that $\vdash_l c\tau$ is derivable, then $c\tau$ normalizes.
\end{theorem}


This is to be contrasted with Okasaki, Lee and Tarditi's semantics for the call-by-need $\lambda$-calculus,
which is not normalizing in the simply-typed case, as shown in Ariola
{\em et al.}~\cite{AriDowHerNakSau12}.

\subsection{Extension to 2$^{\text{nd}}$-order type systems}
We focused in this article on simply-typed versions of the {\lbv} and {\lbvtstar} calculi. But as it is common in Krivine classical realizability,
first and second-order quantifications (in Curry style) come for free through the interpretation. 
This means that we can for instance extend the language of types to first and second-order predicate logic:
$$\begin{array}{rcl}
   e_1,e_2 &::=& x\mid f(e_1,\ldots,e_k)\\
   A,B     &::=& X(e_1,\ldots,e_k)\mid A\to B \mid \forall x. A\mid \forall X. A
  \end{array}$$

  We can then define the following introduction rules for universal quantifications:
$$
\infer[\ffallrrule]{\Gamma \vdash_v v:\forall x.A}{\Gamma \vdash_v v:A & x\notin FV(\Gamma)}
\hsep
\infer[\sfallrrule]{\Gamma \vdash_v v:\forall X.A}{\Gamma \vdash_v v:A & X\notin FV(\Gamma)}
$$
Observe that these rules need to be restricted at the level of strong values, 
just as they are restricted to values in the case of call-by-value\footnote{For further explanation on 
the need for a value restriction in Krivine realizability, we refer the reader to \cite{Munch09} or \cite{Lepigre16}.}.
As for the left rules, they can be defined at any levels, let say the more general $e$:
$$
\infer[\ffalllrule]{\Gamma \vdash_e e:(\forall x.A)^\negt}{\Gamma \vdash_e e:(A[n/x])^\negt }
\hsep
\infer[\sfalllrule]{\Gamma \vdash_e e:(\forall X.A)^\negt}{\Gamma \vdash_e e:(A[B/X])^\negt}
$$
where $n$ is any natural number and $B$ any formula.
The usual (call-by-value) interpretation of the quantification is defined as an intersection over all the possible instantiations of the variables within the model.
We do not wish to enter into too many details\footnote{Once again, we advise 
the interested reader to refer to \cite{Munch09} or \cite{Lepigre16} for further details.}
on this topic here, but first-order variable are to be instantiated by integers, while second order are to be instantiated
by subset of terms at the lower level, \emph{i.e.} closed strong-values in store (which we write $\mathcal{V}_0$):
$$
\tvv{\forall x.A} = \bigcap_{n\in\mathbb{N}} \tvv{A[n/x]} 
\hsep
\tvv{\forall X.A} = \bigcap_{S\in{\mathbb{N}^k\to\mathcal P}(\mathcal{V}_0)} \tvv{A[S/X]}
$$
where the variable $X$ is of arity $k$.
It is then routine to check that the typing rules are adequate with the realizability interpretation.


\section{Conclusion and further work}
\label{s:conclusion}
In this paper, we presented a system of simple types for a call-by-need calculus with control,
which we proved to be safe in that it satisfies subject reduction (Theorem \ref{thm:subject}) and
that typed terms are normalizing (Theorem \ref{thm:normalization}).
We proved the normalization by means of realizability-inspired interpretation of the $\lbvtstar$-calculus.
Incidentally, this opens the doors to the computational analysis (in the spirit of Krivine realizability) 
of classical proofs using control, laziness and shared memory.

In further work, we intend to present two extensions of the present paper.
First, following the definition of the realizability interpretation, 
we managed to type the continuation-and-store passing style translation
for the $\lbvtstar$-calculus (see \cite{AriDowHerNakSau12}).
Interestingly, typing the translation emphasizes its computational content,
and in particular, the store-passing part is reflected in a Kripke forcing-like
manner of typing the extensibility of the store~\cite[Chapter 6]{these}. 

Second, on a different aspect, the realizability interpretation we introduced could be a first step towards new ways of realizing axioms.
In particular, the first author used in his Ph.D. thesis~\cite[Chapter 8]{these} 
the techniques presented in this paper to give a normalization proof for $\text{dPA}^\omega$, 
a proof system developed by the second author~\cite{Herbelin12}.
Indeed, this proof system allows to define a proof for the axiom of dependent choice
thanks to the use of streams that are lazily evaluated, and was lacking a proper normalization proof.

Finally, to determine the range of our technique, it would be natural to investigate the relation
between our framework and the many different presentations of call-by-need calculi (with or without control).
Amongst other calculi, we could cite Chang-Felleisen presentation of call-by-need~\cite{ChaFel12},
Garcia \emph{et al.} lazy calculus with delimited control~\cite{GarEtAl10} or Kesner's recent paper
on normalizing by-need terms characterized by an intersection type system~\cite{Kesner16}.
To this end, we might rely on Pédrot and Saurin's classical by-need~\cite{PedSau16}.
They indeed relate (classical) call-by-need with linear head-reduction from a computational point of view,
and draw the connections with the presentations of Ariola \emph{et al.}~\cite{AriDowHerNakSau12}
and Chang-Felleisen~\cite{ChaFel12}. Ariola \emph{et al.} $\lbv$-calculus
being close to the $\lbvtstar$-calculus (see~\cite{AriDowHerNakSau12} for further details), 
our technique is likely to be adaptable to their framework, and thus to Pédrot and Saurin's system.

\bibliography{biblio-en}

\vlong{
\newpage
\appendix
\section{Subject reduction of the {\lbvtstar}-calculus}
\label{s:sub_red}
We present in this section the proof of subject reduction for
the $\lbvtstar$-calculus (Section \ref{s:lbvtstar}). 
The proof is done by reasoning by induction over the typing derivation,
and relies on the fact that the type system admits a weakening rule.

\begin{lemma}
\label{lm:lbvtstar_weak}
 The following rule is admissible for any level $o$ of the hierarchy $e,t,E,V,F,v,c,l,\tau$:
 $$\infer{\Gamma'\vdash_o o:A}{\Gamma\vdash_o o:A & \Gamma \subseteq \Gamma'}$$
\end{lemma}
\begin{proof}
 Easy induction on typing derivations using the typing rules given in Figure \ref{fig:typing-rules}.
\end{proof}

\begin{ntheorem}{1}
 If $\Gamma \vdash_l c\tau$ and $c\tau\rightarrow c'\tau'$ then $\Gamma\vdash_l c'\tau'$.
\end{ntheorem}
\begin{proof}
By induction over the reduction rules of the $\lbvtstar$-calculus (see Figure \ref{fig:reduction-rules}).
\vspace{-1em}
\prfcase{$\cut{t}{\tmu x.c }\tau                \rightarrow c\tau[x:=t]                           $}
A typing derivation of the closure on the left-hand side is of the form:
$$
\infer{\Gamma \vdash _l \cut{t}{\tmu x.c }\tau}{
  \infer{\Gamma,\Gamma' \vdash_c \cut{t}{\tmu x.c }}{
    \infer{\Gamma,\Gamma' \vdash_t t:A}{\Pi_t}
    &
    \infer{\Gamma,\Gamma' \vdash_e \tmu x.c:A}{
      \infer{\Gamma,\Gamma',x:A \vdash_c c}{\Pi_c}
    }
   &
   \infer{\Gamma\vdash_\tau \tau:\Gamma'}{\Pi_\tau}
   }
}
$$
hence we can derive:
$$
\infer{\Gamma \vdash _l c\tau[x:=t]}{
  \infer{\Gamma,\Gamma',x:A \vdash_c c}{\Pi_c}
  &
  \infer{\Gamma\vdash_\tau \tau[x:=t]:(\Gamma',x:A)}{
    \infer{\Gamma\vdash_\tau \tau:\Gamma'}{\Pi_\tau}
    &
    \infer{\Gamma,\Gamma' \vdash_t t:A}{\Pi_t}
  }
}
$$

\prfcase{$\cut{\mu\alpha.c}{E}\tau              \rightarrow c\tau[\alpha:=E]                      $}
A typing derivation of the closure on the left-hand side is of the form:
$$
\infer{\Gamma \vdash _l \cut{t}{\tmu x.c }\tau}{
  \infer{\Gamma,\Gamma' \vdash_c \cut{\mu \alpha.c }{E}}{
    \infer{\Gamma,\Gamma' \vdash_t \mu \alpha.c:A}{
      \infer{\Gamma,\Gamma',\alpha:A^\negt \vdash_c c}{\Pi_c}
    }
    &
    \infer{\Gamma,\Gamma' \vdash_e E:A^\negt}{
      \infer{\Gamma,\Gamma' \vdash_E E:A^\negt}{\Pi_E}
    }
   }
   &
   \infer{\Gamma\vdash_\tau \tau:\Gamma'}{\Pi_\tau}
}
$$
hence we can derive:
$$
\infer{\Gamma \vdash _l c\tau[\alpha:=E]}{
  \infer{\Gamma,\Gamma',\alpha:A^\negt \vdash_c c}{\Pi_c}
  &
  \infer{\Gamma\vdash_\tau \tau[\alpha:=E]:(\Gamma',\alpha:A^\negt)}{
    \infer{\Gamma\vdash_\tau \tau:\Gamma'}{\Pi_\tau}
    &
    \infer{\Gamma,\Gamma' \vdash_E E:A}{\Pi_E}
  }
}
$$

\prfcase{$\cut{V}{\alpha}\tau[\alpha:=E]\tau'   \rightarrow \cut{V}{E}\tau[\alpha:=E]\tau'        $}
A typing derivation of the closure on the left-hand side is of the form:
$$
{\small
\infer{\Gamma \vdash _l \cut{V}{\alpha}\tau[\alpha:=E]\tau'}{
  \infer{\Gamma,\Gamma_0,\alpha:A^\negt,\Gamma_1 \vdash_c \cut{V}{\alpha}}{
    \infer{\Gamma,\Gamma_0,\alpha:A^\negt,\Gamma_1 \vdash_t V:A}{\Pi_V}
    &
     \infer{\Gamma,\Gamma_0,\alpha:A^\negt,\Gamma_1  \vdash_e \alpha:A^\negt}{
	\infer{\Gamma,\Gamma_0,\alpha:A^\negt,\Gamma_1 \vdash_E \alpha:A^\negt}{
	  \infer{\Gamma,\Gamma_0,\alpha:A^\negt,\Gamma_1 \vdash_F \alpha:A^\negt}{}
      }
     }
    }
   &
   \infer{\Gamma\vdash_\tau \tau[\alpha:=E]\tau' :\Gamma_0,\alpha:A^\negt,\Gamma_1}{
      \infer{\Gamma\vdash_\tau \tau[\alpha:=E] :\Gamma_0,\alpha:A^\negt }{
        \infer{\Gamma \vdash \tau :\Gamma_0}{\Pi_\tau}
  	& 
        \infer{\Gamma,\Gamma_0\vdash_E E:A^\negt}{\Pi_E}
      }
      &
      \Pi_{\tau'}
   }
}
}
$$
where we cheated to compact each typing judgment for $\tau'$ (corresponding to types in $\Gamma_1$) in $\Pi_{\tau'}$. 
Therefore, we can derive:
$$
{\small
\infer{\Gamma \vdash _l \cut{V}{\alpha}\tau[\alpha:=E]\tau'}{
  \infer{\Gamma,\Gamma_0,\alpha:A^\negt,\Gamma_1 \vdash_c \cut{V}{E}}{
    \infer{\Gamma,\Gamma_0,\alpha:A^\negt,\Gamma_1 \vdash_t V:A}{\Pi_V}
    &
     \infer{\Gamma,\Gamma_0,\alpha:A^\negt,\Gamma_1  \vdash_e E:A^\negt}{
	\infer{\Gamma,\Gamma_0,\alpha:A^\negt,\Gamma_1 \vdash_E E:A^\negt}{
	  \Pi_E
      }
     }
    }
   &
   \infer{\Gamma\vdash_\tau \tau[\alpha:=E]\tau' :\Gamma_0,\alpha:A^\negt,\Gamma_1}{
      \infer{\Gamma\vdash_\tau \tau[\alpha:=E] :\Gamma_0,\alpha:A^\negt }{
        \infer{\Gamma \vdash \tau :\Gamma_0}{\Pi_\tau}
  	& 
        \infer{\Gamma,\Gamma_0\vdash_E E:A^\negt}{\Pi_E}
      }
      &
      \Pi_{\tau'}
   }
}
}
$$

\prfcase{$\cut{x}{F}\tau[x:=t]\tau'             \rightarrow \cut{t}{\tmu[x].\cut{x}{F}\tau'}\tau  $}
A typing derivation of the closure on the left-hand side is of the form:
$${\small
\infer{\Gamma \vdash _l \cut{V}{F}\tau[x:=t]\tau'}{
  \infer{\Gamma,\Gamma_0,x:A,\Gamma_1 \vdash_c \cut{x}{F}}{
    \infer{\Gamma,\Gamma_0,x:A,\Gamma_1 \vdash_t x:A}{
      \infer{\Gamma,\Gamma_0,x:A,\Gamma_1 \vdash_V x:A}{}
    }
    &
     \infer{\Gamma,\Gamma_0,x:A,\Gamma_1  \vdash_e F:A^\negt}{\Pi_F
     }
    }
   &
   \infer{\Gamma\vdash_\tau \tau[ x:=t]\tau' :\Gamma_0, x:A,\Gamma_1}{
      \infer{\Gamma\vdash_\tau \tau[ x:=t] :\Gamma_0, x:A }{
        \infer{\Gamma \vdash \tau :\Gamma_0}{\Pi_\tau}
  	& 
        \infer{\Gamma,\Gamma_0\vdash_t t:A}{\Pi_t}
      }
      &
      \Pi_{\tau'}
   }
}
}
$$
hence we can derive:
$${\small
\infer{\Gamma \vdash _l \cut{t}{\tmu[x].\cut{x}{F}\tau'}\tau}{
  \infer{\Gamma,\Gamma_0 \vdash_c \cut{t}{\tmu[x].\cut{x}{F}\tau'}}{
    \infer{\Gamma,\Gamma_0,\Gamma_1 \vdash_t t:A}{
      \Pi_t
    }
    &\hspace{-1cm}
     \infer{\Gamma,\Gamma_0  \vdash_e \tmu[x].\cut{x}{F}\tau':A^\negt}{
      \infer{\Gamma,\Gamma_0  \vdash_E \tmu[x].\cut{x}{F}\tau':A^\negt}{
 	\infer{\Gamma,\Gamma_0,x:A\vdash_l  \cut{x}{F}\tau' }{
 	  \infer{\Gamma,\Gamma_0,x:A,\Gamma_1 \vdash_c \cut{x}{F}}{
 	    \infer{\Gamma,\Gamma_0,x:A,\Gamma_1 \vdash_t x:A}{
	      \infer{\Gamma,\Gamma_0,x:A,\Gamma_1 \vdash_V x:A}{}
	    }
 	    &
 	    \infer{\Gamma,\Gamma_0,x:A,\Gamma_1 \vdash_e F:A^\negt}{\Pi_F}
 	  }
	  &
	  \infer{\Gamma,\Gamma_0,x:A\vdash_\tau \tau':\Gamma_1}{\Pi_{\tau'}}
	}
      }
    }
  }
  &\hspace{-2cm}
  \infer{\Gamma \vdash \tau :\Gamma_0}{\Pi_\tau}
}
}
$$

\prfcase{$\cut{V}{\tmu[x].\cut{x}{F}\tau'}\tau  \rightarrow \cut{V}{F}\tau[x:=V]\tau'             $}
A typing derivation of the closure on the left-hand side is of the form:
$$
{\small
\infer{\Gamma \vdash _l \cut{V}{\tmu[x].\cut{x}{F}\tau'}\tau}{
  \infer{\Gamma,\Gamma_0 \vdash_c \cut{V}{\tmu[x].\cut{x}{F}\tau'}}{
    \infer{\Gamma,\Gamma_0,\Gamma_1 \vdash_t V:A}{
      \Pi_V
    }
    &\hspace{-1cm}
     \infer{\Gamma,\Gamma_0  \vdash_e \tmu[x].\cut{x}{F}\tau':A^\negt}{
      \infer{\Gamma,\Gamma_0  \vdash_E \tmu[x].\cut{x}{F}\tau':A^\negt}{
 	\infer{\Gamma,\Gamma_0,x:A\vdash_l  \cut{x}{F}\tau' }{
 	  \infer{\Gamma,\Gamma_0,x:A,\Gamma_1 \vdash_c \cut{x}{F}}{
 	    \infer{\Gamma,\Gamma_0,x:A,\Gamma_1 \vdash_t x:A}{
	      \infer{\Gamma,\Gamma_0,x:A,\Gamma_1 \vdash_V x:A}{}
	    }
 	    &
 	    \infer{\Gamma,\Gamma_0,x:A,\Gamma_1 \vdash_e F:A^\negt}{\Pi_F}
 	  }
	  &
	  \infer{\Gamma,\Gamma_0,x:A\vdash_\tau \tau':\Gamma_1}{\Pi_{\tau'}}
	}
      }
    }
  }
  &\hspace{-2cm}
  \infer{\Gamma \vdash \tau :\Gamma_0}{\Pi_\tau}
}
}
$$
Therefore we can derive:
$$
{\small
\infer{\Gamma \vdash _l \cut{V}{F}\tau[x:=V]\tau'}{
  \infer{\Gamma,\Gamma_0,x:A,\Gamma_1 \vdash_c \cut{V}{F}}{
    \infer{\Gamma,\Gamma_0,x:A,\Gamma_1 \vdash_t V:A}{
      \Pi_V
    }
    &
     \infer{\Gamma,\Gamma_0,x:A,\Gamma_1  \vdash_e F:A^\negt}{\Pi_F
     }
    }
   &
   \infer{\Gamma\vdash_\tau \tau[ x:=V]\tau' :\Gamma_0, x:A,\Gamma_1}{
      \infer{\Gamma\vdash_\tau \tau[ x:=V] :\Gamma_0, x:A }{
        \infer{\Gamma \vdash \tau :\Gamma_0}{\Pi_\tau}
  	& 
        \infer{\Gamma,\Gamma_0\vdash_t V:A}{\Pi_V}
      }
      &
      \Pi_{\tau'}
   }
}
}
$$
where we implicitly used Lemma~\ref{lm:lbvtstar_weak} to weaken $\Pi_V$:
$$
\infer{\Gamma,\Gamma_0,x:A,\Gamma_1\vdash_t V:A}{
  \infer{\Gamma,\Gamma_0\vdash_t V:A}{\Pi_V}
  & 
  \Gamma,\Gamma_0 \subseteq \Gamma,\Gamma_0,x:A,\Gamma_1
}
$$

\prfcase{$\cut{\lambda x.t}{u\cdot E}\tau       \rightarrow \cut{u}{\tmu x.\cut{t}{E}}\tau        $}

A typing proof for the closure on the left-hand side is of the form:
{
$$
\infer{\Gamma\vdash_l \cut{\lambda x.t}{u\cdot E}\tau}{
 \infer{\Gamma,\Gamma'\vdash_c \cut{\lambda x.t}{u\cdot E}}{
   \infer{\Gamma,\Gamma' \vdash_t \lambda x.t : A \to B}{
   \infer{\Gamma,\Gamma' \vdash_V \lambda x.t : A \to B}{
   \infer{\Gamma,\Gamma' \vdash_v \lambda x.t : A \to B}{
     \infer{\Gamma,\Gamma',x:A \vdash_t t : B }{\Pi_t}
   }
   }
   }
   &
   \infer{\Gamma,\Gamma' \vdash_e u\cdot e : (A\to B)^\negt}{
     \infer{\Gamma,\Gamma' \vdash_E u\cdot e : (A\to B)^\negt}{
       \infer{\Gamma,\Gamma' \vdash_t u: A}{\Pi_u}
       & 
       \infer{\Gamma,\Gamma' \vdash_E E : B^\negt}{\Pi_E}
     }
   }
  }
  &
  \infer{\Gamma\vdash_\tau \tau:\Gamma'}{\Pi_\tau}
}
$$
}
We can thus build the following derivation:

$$
\infer{\Gamma \vdash_l \cut{u}{\tmu x .\cut{t}{E}}\tau}{
  \infer{\Gamma,\Gamma' \vdash_c \cut{u}{\tmu x .\cut{t}{E}}}{
    \infer{\Gamma,\Gamma' \vdash_t u: A}{\Pi_u}
    &
    \infer{\Gamma,\Gamma'\vdash_e \tmu x.\cut{t}{E}:A^\negt}{
      \infer{\Gamma,\Gamma',x:A\vdash_c \cut{t}{E}}{
        \infer{\Gamma,\Gamma',x:A \vdash_t t : B }{\Pi_t}
        &
        \infer{\Gamma,\Gamma',x:A \vdash_e E : B^\negt }{
  	\infer{\Gamma,\Gamma',x:A \vdash_E E : B^\negt }{\Pi_E}
        }
      }
     }
  }
  &
  \infer{\Gamma\vdash_\tau \tau:\Gamma'}{\Pi_\tau}
}
$$
where we implicitly used Lemma~\ref{lm:lbvtstar_weak} to weaken $\Pi_E$:
$$
\infer{\Gamma,\Gamma',x:A\vdash_E E:B^\negt}{
  \infer{\Gamma,\Gamma\vdash_E E:B^\negt}{\Pi_E}
  & 
  \Gamma,\Gamma' \subseteq \Gamma,\Gamma',x:A
}
$$
 
\end{proof}

\section{Adequacy lemma}

\label{a:real}
We give here the full proof of the adequacy lemma for 
the realizability interpretation of the $\lbvtstar$-calculus.

\begin{ntheorem}{17}[Adequacy]
The typing rules of \Cref{fig:typing-rules} for the $\lbvtstar$-calculus without co-constants
are adequate with any pole. 
Namely, if $\Gamma$ is a typing context, $\pole$ is a pole and $\tau$ is a store such that ${\tau\real \Gamma}$,
then the following holds in the $\lbvtstar$-calculus without co-constants:
\begin{enumerate}
 \item If $v$ is a strong value 	such that $\Gamma\vdash_v v:A$, 	then $\tis{v}{\tau} \in\tvv{A}$.
 \item If $F$ is a forcing context 	such that $\Gamma\vdash_F F:A^\negt$, 	then $\tis{F}{\tau} \in\fvF{A}$.
 \item If $V$ is a weak value   	such that $\Gamma\vdash_V V:A$, 	then $\tis{V}{\tau} \in\tvV{A}$.
 \item If $E$ is a catchable context 	such that $\Gamma\vdash_E E:A^\negt$, 	then $\tis{E}{\tau} \in\fvF{A}$.
 \item If $t$ is a term    		such that $\Gamma\vdash_t t:A$, 	then $\tis{t}{\tau} \in\tvt{A}$.
 \item If $e$ is a context 		such that $\Gamma\vdash_e e:A^\negt$, 	then $\tis{e}{\tau} \in\fve{A}$.
 \item If $c$ is a command		such that $\Gamma\vdash_c c$, 		then $c\tau \in \pole$. 
 \item If $\tau'$ is a store	such that $\Gamma\vdash_\tau \tau':\Gamma'$, 	then $\tau\tau' \real \Gamma,\Gamma'$. 
\end{enumerate}
\end{ntheorem}
\begin{proof}
\renewcommand{\prfcasetext}{Rule}
We proceed by induction over the typing rules.

\prfcase{\cboldrule}
This case stems directly from the definition of $\tvv{X}$ for $X$ atomic.

\prfcase{\imprrule}
This case exactly matches the definition of $\tvv{A\imp B}$.
Assume that $$\infer[\imprrule]{\Gamma\vdash_v\lambda x.t:A\imp B}{\Gamma,x:A\vdash_t t:B}$$ 
and let $\pole$ be a pole and $\tau$ a store such that $\tau \real \Gamma$.
If $\tis{u}{\tau'}$ is a closed term in the set $\tvt{A}$, then, up to $\alpha$-conversion
for the variable $x$, $\overline{\tau\tau'} \real \Gamma$ by Lemma~\ref{lm:st_weak}
and $\overline{\tau\tau'}[x:=u] \real \Gamma,x:A$.
Using the induction hypothesis, $\tis{t}{\overline{\tau\tau'}[x:=u]}$ is indeed in $\tvt{B}$.

\prfcase{\implrule}
Assume that $$\infer[\implrule]{\Gamma\vdash_F u\cdot E:(A\imp B)^\negt}{\Gamma\vdash_t u:A & \Gamma \vdash_E E:B^\negt}$$ 
and let $\pole$ be a pole and $\tau$ a store such that $\tau \real \Gamma$.
Let $\tis{\lambda x.t}{\tau'}$ be a closed term in the set $\tvv{A\to B}$ such that $\compat{\tau}{\tau'}$, then we have:
$$\cut{\lambda x.t}{u\cdot E}\overline{\tau\tau'}
~~\rightarrow ~~\cut{u}{\tmu x.\cut{t}{E}}\overline{\tau\tau'}
~~\rightarrow ~~\cut{t}{E}\overline{\tau\tau'}[x:=u]$$
By definition of $\tvv{A\to B}$, this closure is in the pole, and we can conclude by anti-reduction.

\prfcase{\liftVrule} This case, as well as every other case where typing a term (resp. context)
at a higher level of the hierarchy (rules \liftErule, \lifttrule, \lifterule), is a simple consequence of 
Proposition~\ref{prop:monotonicity}.
Indeed, assume for instance that 
$$\infer[\liftVrule]{\Gamma\vdash_V v:A}{\Gamma\vdash_v v:A}$$ 
and let $\pole$ be a pole and $\tau$ a store such that $\tau \real \Gamma$.
By induction hypothesis, we get that $\tis{v}{\tau}\in\tvv{A}$.
Thus if $\tis{F}{\tau'}$ is in $\fvF{A}$, by definition $\tis{v}{\tau}\orth\tis{F}{\tau'}$.

\prfcase{\xrule} 
Assume that 
$$\infer[\xrule]{\Gamma\vdash_V x:A}{(x:A)\in\Gamma}$$ 
and let $\pole$ be a pole and $\tau$ a store such that $\tau \real \Gamma$.
As $(x:A)\in\Gamma$, we know that $\tau$ is of the form $\tau_0[x:=t]\tau_1$
with $\tis{t}{\tau_0}\in\tvt{A}$.
Let $\tis{F}{\tau'}$ be in $\fvF{A}$, with  $\compat{\tau}{\tau'}$. By Lemma~\ref{lm:st_union},
we know that $\overline{\tau\tau'}$ is of the form $\overline{\tau_0}[x:=t]\overline{\tau_1}$.
Hence we have:
$$\cut{x}{F} \overline{\tau_0}[x:=t]\overline{\tau_1} ~~\rightarrow~~ \cut{t}{\tmu[x].\cut{x}{F}\overline{\tau_1}}\overline{\tau_0}$$
and it suffices by anti-reduction to show that the last closure is in the pole $\pole$. 
By induction hypothesis, we know that $\tis{t}{\tau_0}\in\tvt{A}$ thus we only need to show 
that it is in front of a catchable context in $\fvE{A}$.
This corresponds exactly to the next case that we shall prove now.

\prfcase{\eagerrule} 
Assume that
$$\infer[\eagerrule]{\Gamma\vdash_E \tmu[x].\cut{x}{F}\tau':A}{\Gamma,x:A,\Gamma'\vdash_F F:A & \Gamma,x:A\vdash \tau':\Gamma'}$$ 
and let $\pole$ be a pole and $\tau$ a store such that $\tau \real \Gamma$.
Let $\tis{V}{\tau_0}$ be a closed term in $\tvV{A}$ such that $\compat{\tau_0}{\tau}$.
We have that :
$$\cut{V}{\tmu[x].\cut{x}{F}\overline{\tau'}}\overline{\tau_0\tau}~~\rightarrow~~ \cut{V}{F} \overline{\tau_0\tau}[x:=V]\tau'$$
By induction hypothesis, we obtain $\tau[x:=V]\tau'\real \Gamma,x:A,\Gamma'$. 
Up to $\alpha$-conversion in $F$ and $\tau'$, so that the variables in $\tau'$ are disjoint from those in $\tau_0$,
we have that $\overline{\tau_0\tau}\real\Gamma$ (by Lemma~\ref{lm:st_weak}) and then $\tau''\defeq\overline{\tau_0\tau}[x:=V]\tau'\real \Gamma,x:A,\Gamma'$.
By induction hypothesis again, we obtain that $\tis{F}{\tau''}\in\fvF{A}$ (this was an assumption in the previous case) and
as $\tis{V}{\tau_0}\in\tvV{A}$, we finally get that $\tis{V}{\tau_0}\orth\tis{F}{\tau''}$ 
and conclude again by anti-reduction.

\prfcases{\alpharule}
This case is obvious from the definition of $\tau\real\Gamma$.

\prfcase{\murule}
Assume that
$$\infer[\murule]{\Gamma\vdash_t \mu\alpha.c:A}{\Gamma,\alpha:A^\negt\vdash_c c}$$ 
and let $\pole$ be a pole and $\tau$ a store such that $\tau \real \Gamma$.
Let $\tis{E}{\tau'}$ be a closed context in $\fvE{A}$ such that $\compat{\tau}{\tau'}$.
We have that :
$$\cut{\mu\alpha.c}{E}\overline{\tau\tau'}~~\rightarrow~~ c \overline{\tau\tau'}[\alpha:=E]$$
Using the induction hypothesis, we only need to show that $\overline{\tau\tau'}[\alpha:=E]\real \Gamma,\alpha:A^\negt,\Gamma'$ 
and conclude by anti-reduction.
This obviously holds, since $\tis{E}{\tau'}\in\fvE{A}$ and $\overline{\tau\tau'}\real\Gamma$ byLemma~\ref{lm:st_union}.

\prfcase{\mutrule}
This case is identical to the previous one.

\prfcase{\crule}
Assume that
$$\infer[\crule]{\Gamma\vdash_c \cut{t}{e}}{\Gamma\vdash_t t:A & \Gamma \vdash_e e:A^\negt}$$ 
and let $\pole$ be a pole and $\tau$ a store such that $\tau \real \Gamma$.
Then by induction hypothesis $\tis{t}{\tau}\in\tvt{A}$ and $\tis{e}{\tau}\in\fve{A}$, so that
$\cut{t}{e}\tau\in\pole$.

\prfcase{\tautrule}
This case directly stems from the induction hypothesis which exactly matches the definition of {$\tau\tau'[x:=t]\real\Gamma,\Gamma',x:A$}. 
The case for the rule {\tauErule} is identical, and the case for the rule {\epsrule} is trivial.
 \end{proof}

}
\end{document}